\documentclass{sig-alternate-05-2015}

\usepackage[usenames]{color}
\usepackage{epsfig}
\usepackage{epstopdf}
\usepackage{url}

\newfont{\mycrnotice}{ptmr8t at 7pt}
\newfont{\myconfname}{ptmri8t at 7pt}
\let\crnotice\mycrnotice%
\let\confname\myconfname%

\setcopyright{acmcopyright}

\begin{document}

\title{Scaled VIP Algorithms for Joint Dynamic Forwarding and Caching in Named Data Networks}

\numberofauthors{2} 
%
\author{
%
%
\alignauthor Fan Lai,  Feng Qiu,  Wenjie Bian, Ying Cui\thanks{Y.~Cui gratefully acknowledges support from the National Science Foundation of China grant 61401272. }\\
       \affaddr{Shanghai Jiao Tong University}\\
      \affaddr{Shanghai, China}\\
       \alignauthor
       Edmund Yeh\thanks{E.~Yeh gratefully acknowledges support from the National Science Foundation Future Internet Architecture grant CNS-1205562 and a Cisco Systems research grant. }\\
       \affaddr{Northeastern University}\\
       \affaddr{Boston, MA, USA}\\
%
}

\maketitle

\newtheorem{Thm}{Theorem}
\newtheorem{Alg}{Algorithm}
\newtheorem{Def}{Definition}
\newtheorem{Rem}{Remark}

\allowdisplaybreaks
\begin{abstract}
Emerging Information-Centric Networking (ICN) architectures seek to optimally utilize both bandwidth and storage for efficient content distribution over the network.  The Virtual Interest Packet (VIP) framework has been proposed to enable joint design of forwarding and caching within the Named Data Networking (NDN) architecture.  The virtual plane of the VIP framework captures the measured demand for content objects, but does not reflect interest collapse and suppression in the NDN network.  We aim to further improve the  performance of the existing VIP algorithms by using a modified virtual plane where VIP counts are appropriately scaled to reflect interest suppression effects.  We characterize the stability region of the modified virtual plane with VIP scaling, develop a new distributed forwarding and caching algorithm operating on the scaled VIPs, and demonstrate the throughput optimality  of the scaled VIP algorithm in the virtual plane.
Numerical experiments demonstrate significantly enhanced performance relative to the existing VIP algorithm, as well as a number of other baseline algorithms.
\end{abstract}

\category{C.2.1}{Computer-Communication Networks}{Network Architecture and Design---network communications}

\terms{Theory, Design, Management}

\keywords{Named data networking, content centric networking, information centric networking, forwarding, caching, interest collapse}




\section{Introduction}


It is increasingly clear that traditional connection-based networking architectures are ill suited for the prevailing user demands for network content~\cite{Zhang10}.
Emerging Information-Centric Networking (ICN) architectures aim to remedy this fundamental mismatch so as to dramatically improve the efficiency of content dissemination over the Internet.
In particular, Named Data Networking (NDN)~\cite{Zhang10}, or Content-Centric Networking (CCN) \cite{Jacobson},
is a proposed network
architecture for the Internet that replaces the traditional client-server  communication model   with
one based on the identity of data or content.

 Content delivery in NDN is accomplished using {\em Interest Packets} and {\em Data Packets}, along with specific data structures in nodes such as the {\em Forwarding Information Base (FIB)}, the {\em Pending Interest Table (PIT)}, and the {\em Content Store} (cache).
 Communication is initiated by a data consumer or requester sending a request for the data using an \emph{Interest Packet}. Interest Packets are   forwarded along routes determined by the FIB at each node.
Repeated requests for the same object can be suppressed at each node according to its PIT.
The {\em Data Packet} is subsequently transmitted back along the path taken by the corresponding Interest Packet, as recorded by the PIT at each node.
A node may optionally cache the data objects contained in the  received Data Packets in its local {\em Content Store}.
Consequently, a request for a data object can be fulfilled not only by the content source but also by any node with a copy of that object in its cache.
Please see~\cite{Zhang10,VIPICN14} for details.

NDN seeks to optimally utilize both bandwidth and storage for efficient content distribution, which highlights the need for joint design of traffic engineering and caching strategies, in order to optimize network performance.
To address this fundamental problem, \cite{VIPICN14} proposes the {\em VIP framework}  for the design of high performing NDN networks.
 Within this VIP framework,  joint dynamic forwarding and caching  algorithms operating on virtual interest packets (VIPs)    are developed to maximize network stability in the virtual plane, using Lyapunov drift techniques \cite{VIPICN14}. The joint dynamic algorithms in \cite{VIPICN14} extend the existing  dynamic backpressure algorithm in the sense that they incorporate caching into the dynamic design.  Then, using the resulting flow rates and queue lengths of the VIPs in the virtual plane, \cite{VIPICN14} develops corresponding  joint dynamic  algorithms  in the actual plane, which have been shown to achieve superior performance in terms of   user delay and cache hit rates, relative to several baseline policies.


The virtual plane of the VIP framework in \cite{VIPICN14} focuses on capturing the measured demand for content objects in the network, but does not reflect interest collapse and suppression in the NDN network. In this paper, we aim to further improve the  performance of the existing VIP algorithms  in \cite{VIPICN14} by reflecting more accurately the actual interest packet traffic in the NDN network under interest suppression.   There are several potential challenges in pursuing this. First, it is not clear how one should improve the delay performance of the existing VIP algorithms by simultaneously capturing the network demand and the interest suppression effect, in a tractable manner. Second, it is not clear how to maintain the desired throughput optimality of
the existing VIP algorithms when the Lyapunov-drift-based control structure is modified to reflect interest suppression in the actual network.

 In the following, we shall
address the above questions and challenges.
We first  propose  a modified virtual plane where VIP counts are appropriated {\em scaled}.  In a simple manner, the scaled VIP counts capture to some extent the measured demand, but also reflect to some extent interest suppression effects.  We characterize the stability region of the modified virtual plane with VIP scaling, which is a superset of the stability region in \cite{VIPICN14} without VIP scaling.  We then develop a new distributed forwarding and caching algorithm operating on the scaled VIPs.  Generalizing Lyapunov drift techniques, we demonstrate the throughput optimality  of the scaled VIP algorithm in the virtual plane.  The scaled VIP algorithm generalizes the VIP algorithm  in \cite{VIPICN14}.
Numerical experiments   demonstrate the superior performance of the resulting stable VIP algorithm for handling  Interest Packets and Data Packets  within the actual plane, in terms of low network delay,  relative to  a number of baseline alternatives.

Although there is now a rapidly growing literature in ICN, the problem of optimal joint forwarding and caching for content-oriented networks remains challenging.  In \cite{RossiniICN14}, the authors demonstrate the gains of joint forwarding and caching in ICNs. 
In \cite{CachingTM2011:Ying},   the authors propose throughput optimal one-hop routing and caching   in a single-hop Content Distribution Network (CDN).

Throughput-optimal caching and routing in  multi-hop networks remains an open problem.
In \cite{TECC2012},  assuming the path between any two nodes is predetermined,
the authors consider single-path routing and caching to minimize link utilization for a general multi-hop content-oriented network.
The benefits of selective caching based on the concept of betweenness centrality, relative to ubiquitous caching, are shown in~\cite{Chai:2012:CLM:2342042.2342046}.
In~\cite{Age-based:6193504}, cooperative caching schemes  have been heuristically designed without being jointly optimized with forwarding strategies.  Finally, adaptive multi-path forwarding in NDN has been examined in~\cite{Yi:2012:AFN:2317307.2317319}, but has not been jointly optimized with caching.

\section{Network Model}\label{sec:model}

 We consider the same network model as in \cite{VIPICN14}, which we describe for completeness.
{Consider a connected multi-hop (wireline) network modeled by a directed graph $\mathcal G=(\mathcal N, \mathcal L)$, where $\mathcal N$ and $\mathcal L$ denote the sets of $N$ nodes and $L$ directed links, respectively.  Assume  that $(b,a) \in {\cal L}$ whenever $(a,b) \in {\cal L}$.  Let $C_{ab} > 0$ be the transmission capacity (in bits/second) of link $(a,b) \in {\cal L}$.  Let $L_n\geq 0$ be the cache size (in bits) at node $n \in {\cal N}$.


Assume that content in the network are identified as {\em data objects},  each consisting of multiple data chunks.
Content delivery in NDN operates at the level of data chunks.  That is, each Interest Packet requests a particular data chunk, and a matching Data Packet consists of the requested data chunk, the data chunk name, and a signature.    A request for a data object consists of a sequence of Interest Packets which request all the data chunks of the object.
We consider a set ${\cal K}$ of $K$ data objects, which may be determined by the amount of control state that the network is able to maintain, and
 may include only the most popular data objects in the network, typically responsible for most of the network congestion.\footnote{The less popular data objects not in ${\cal K}$ may be distributed using simple techniques such as shortest-path forwarding with little or no caching.}
For simplicity, we assume that
all data objects have the same size $D$ (in bits).   The results in the paper can be
extended to the more general case where object sizes differ.\footnote{For the general case with different object sizes, $C_{ba}/D$ in \eqref{eqn:rout_cost_sum} and \eqref{eqn:forwarding-VIP} are replaced with $C_{ba}$,  $A^k_n(t)$ in \eqref{eqn:queue_dyn} is replaced with $A^k_n(t)D^k$, and the constraint in \eqref{eqn:knapsack-VIP} is replaced with $\sum_{k\in \mathcal K} s^k_n D^k \leq L_n$, where $D^k$ denotes the size of object $k$.}   We consider
the scenario where  $L_n < KD$ for all $n \in
{\cal N}$.  Thus, no node can cache all data objects.
For each data object $k \in {\cal K}$, assume
that there is a unique node $src(k) \in {\cal N}$ which serves as the
content source for the object.    Interest Packets for chunks of a given data object
can enter the network at any node, and exit the network upon being satisfied
by matching Data Packets at the content source for the object, or at the
nodes which decide to cache the object.
For convenience, we assume that the content sources are fixed, while the caching points may vary in time.

\section{VIP Framework}\label{sec:vipframe}

\begin{figure}[t]
\begin{center}
\includegraphics[width=60mm,height=20mm]{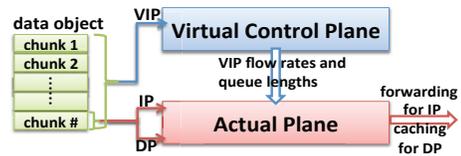}
\caption{VIP framework \cite{VIPICN14}.  IP (DP) stands for Interest Packet (Data Packet).}\label{fig:planes}
\end{center}
\end{figure}

We first review  the VIP framework proposed in \cite{VIPICN14} to facilitate the discussion of
the algorithms developed in later sections. Please refer to \cite{VIPICN14} for the details on the motivation and utility of this framework. As illustrated in Fig.~\ref{fig:planes},  the VIP framework relies on virtual interest packets (VIPs), which capture the {\em measured demand} for the respective data objects (i.e., represent content popularity which is empirically measured, rather than being given a priori). Note that the demand is unavailable in the interior of the actual network due to interest collapsing and suppression.
The VIP framework employs a \emph{virtual} control plane  operating on VIPs {\em at the data object level}, and an \emph{actual} plane  handling Interest Packets and Data Packets {\em at the data chunk level}.   The  virtual plane facilitates the design of  distributed control algorithms operating on VIPs, aimed at yielding desirable performance in terms of network metrics of concern,  by taking advantage of local information on network demand (as represented by the VIP counts).
The flow rates and queue lengths of the VIPs resulting from the control algorithm in the virtual plane are then used to specify the control algorithms in the actual plane \cite{VIPICN14}.

 While the VIP counts in \cite{VIPICN14} capture the measured demand for the respective data objects in the network, they do not reflect interest collapsing and suppression in the actual network.  To provide potentially better guidance for designing efficient control algorithms in the actual plane, we propose  a modified virtual plane operating on VIPs with dynamics different from those in \cite{VIPICN14}.
The VIP counts in the modified virtual plane capture to some extent the measured demand, but also reflect to some extent the actual interest packet traffic under interest collapse and suppression. 

We now specify the  modified dynamics of the VIPs within the  virtual plane.
Consider time slots of length 1 (without loss of generality) indexed by $t = 1, 2, \ldots$.  Specifically, time slot $t$ refers to the time interval $[t,t+1)$.
Within the virtual plane, each node $n\in {\cal N}$ maintains a separate VIP queue  for each data object $k \in {\cal K}$.
Note that no data is contained in these VIPs.  Thus, the VIP queue size for each node $n$ and data object $k$ at the beginning of slot $t$   is  represented by a {\em counter} $V_n^k(t)$.\footnote{\small{We assume that VIPs can be quantified as a real number.  This is reasonable when the VIP counts are large.}}
An exogenous request for data object $k$ is considered to have arrived at node $n$ if the Interest Packet requesting the starting chunk of data object $k$ has arrived at node $n$.
Let $A^k_n(t)$ be the number of exogenous data object request arrivals at node $n$ for object $k$ during slot $t$.\footnote{\small{We think of a node  as an aggregation  point  combining many network users, 
and hence it is likely to submit many requests for a data object over time.}}
For every arriving request for data object $k$ at node $n$, a corresponding VIP for object $k$ is generated at $n$.
The long-term exogenous VIP arrival rate at node $n$ for object $k$ is
$ \lambda_n^k\triangleq \mathbb \limsup_{t \rightarrow \infty} \frac{1}{t} \sum_{\tau = 1}^t A^k_n(\tau).$
Let $\mu_{ab}^k(t)  \geq 0$ be the allocated transmission rate of VIPs for data object $k$ over link $(a,b)$ during time slot $t$.  Note that
a single message between node $a$ and node $b$ can summarize all the VIP transmissions during  each slot. 
Data Packets for the requested data object must travel on the reverse path taken by the Interest Packets.  Thus, in determining the transmission of the VIPs, we  consider the link capacities on the reverse path below:
\begin{align}
&\sum_{k \in \mathcal K} {\mu^{k}_{ab}(t)} \leq C_{ba}/D,\
\text{for~all}~(a,b) \in \mathcal L\label{eqn:rout_cost_sum}\\
&\mu^{k}_{ab}(t)=0, \;\text{for~all}~(a,b)\not\in \mathcal
L^{k}\label{eqn:rout_cost_non_neg}
\end{align}
where $C_{ba}$ is the capacity of ``reverse" link $(b,a)$ and  $\mathcal L^k$ is the set of $L^k$ links which are allowed to transmit the VIPs of object $k$.  Let $C_{\max}\triangleq \max_{(a,b)\in\mathcal L}C_{ab}/D$.

In the virtual plane, we may assume that at each slot $t$, each node $n \in {\cal N}$ can gain access to any data object $k \in {\cal K}$ for which there is interest at $n$, and potentially cache the object locally.
Let $s_n^{ k}(t) \in \{0,1\}$ represent the caching state for object $k$ at node $n$ during slot $t$, where $s_n^{k}(t)=1$ if object $k$  is cached at node $n$ during slot $t$, and $s_n^{k}(t)=0$ otherwise.  Note that even if $s_n^{k}(t)=1$, the content store at node $n$ can satisfy only a limited number of VIPs during one time slot.  This is because there is a maximum rate $r_n$ (in objects per slot) at which node $n$ can produce copies of cached object $k$ \cite{VIPICN14}.


In contrast to \cite{VIPICN14}, in order to reflect network demand as well as interest  suppression in the actual network in a simple manner, we consider {\em scaled VIP counts} in the virtual plane.  Specifically, for object $k$ at node $n$, we scale down the total exogenous and endogenous VIP arrivals $A^k_n(t)+
\sum_{a\in \mathcal N}\mu^{k}_{an}(t)$ at slot $t$ by a constant $\theta_n^k\geq 1$. Note that $\theta_n^k$ is a design parameter for the virtual plane, and can be chosen to reflect the average interest suppression effect  within a slot for object $k$ at node $n$.

Initially, all VIP counters are set to 0, i.e., $V_n^k(1)=0$.
The time evolution of the VIP count at node $n$ for object $k$ is as follows:
\begin{small}
\begin{align}
& V^k_n(t+1) \leq \nonumber\\
&   \left(
\left(V^k_n(t)-\sum_{b\in \mathcal N}\mu^{k}_{nb}(t)\right)^+ +\dfrac{A^k_n(t)+
\sum_{a\in \mathcal N}\mu^{k}_{an}(t)}{\theta_n^k}- r_n s_n^{k}(t)\right)^+
\label{eqn:queue_dyn}
\end{align}
\end{small}
where $(x)^+ \triangleq \max(x,0)$. Note that when $\theta_n^k=1$ for all $k\in\mathcal K$ and $n\in\mathcal N$, the dynamics of VIPs in \eqref{eqn:queue_dyn} reduces to the VIP dynamics in \cite{VIPICN14}. Thus, VIP counts in \cite{VIPICN14} capture the  demand for the respective data objects without interest suppression. By setting $\theta_n^k\geq 1$ for all $k\in\mathcal K$ and $n\in\mathcal N$, VIP counts here can reflect the demand for the respective data objects with interest scaling.

From~\eqref{eqn:queue_dyn}, it can be seen that the VIPs for data object $k$ at node $n$ at the beginning of slot $t$ are transmitted during slot $t$ at the rate $\sum_{b\in \mathcal N}\mu^{k}_{nb}(t)$.  The exogenous and endogenous VIP arrivals  $A^k_n(t)+
\sum_{a\in \mathcal N}\mu^{k}_{an}(t)$ during slot $t$ are scaled down by $\theta_n^k$.
The remaining VIPs $(V^k_n(t)-\sum_{b\in \mathcal N}\mu^{k}_{nb}(t))^+$, as well as the scaled-down  exogenous and endogenous VIP arrivals $ (A^k_n(t)+
\sum_{a\in \mathcal N}\mu^{k}_{an}(t))/{\theta_n^k}$ during slot $t$, are reduced by $r_n$ at the end of slot $t$ if object $k$ is cached at node $n$ in slot $t$ ($s_n^k(t) = 1$). The VIPs still remaining are then transmitted during the next slot $t+1$.
Note that~\eqref{eqn:queue_dyn} is an inequality because the actual number of VIPs for object $k$
arriving to node $n$ during slot $t$ may be less than $\sum_{a\in \mathcal N}\mu^{k}_{an}(t)$ if the neighboring nodes have little or
no VIPs of object $k$ to transmit.
Furthermore,
$V^k_n(t) = 0$ for all $t \geq 1$ if $n=src(k)$.  In other words, the VIPs for an object exit the network once they reach the source node of the object.
Physically, the VIP count can be interpreted as a {\em potential}.  For any  data object, there is a downward ``gradient" from entry points of the data object requests to the content source and caching nodes.


The VIP queue at node $n$ is  {\em stable} if
$$ \limsup_{t \rightarrow \infty}
\frac{1}{t} \sum_{\tau = 1}^{t} 1_{[V^k_{n}(\tau) > \xi]} d\tau
\rightarrow 0 \;\; \text{as} \;\; \xi \rightarrow \infty,$$ where $1_{\{\cdot\}}$
is the indicator function. The {\em VIP network stability region} $\Lambda$ is the closure
of the set of all VIP arrival rates $\boldsymbol \lambda\triangleq (\lambda^k_n)_{n \in {\cal N},k \in {\cal K}}$ for which there exists some feasible (i.e., satisfying~\eqref{eqn:rout_cost_sum}-\eqref{eqn:rout_cost_non_neg} and the cache size limits $(L_n)_{n \in {\cal N}}$)  joint forwarding and caching policy which can guarantee that all VIP queues are stable \cite{VIPICN14}.
Assume (i) the VIP arrival processes $\{A^k_n(t): t=1,2,\ldots\}$  are mutually independent with respect to $n$ and $k$; (ii)
for all $n $ and $k  $, $\{A^k_n(t): t = 1, 2, \ldots\}$ are i.i.d. with respect to $t$; and (iii) for all $n$ and $k$, $A^k_n(t)\leq A^k_{n,\max}$ for all $t$.
The theoretical results in this paper  hold under these assumptions.
We now characterize the VIP stability region in the modified virtual plane (with VIP scaling).

\begin{Thm} [Scaled VIP Stability Region] The $ $ $ $ VIP stability region of the network ${\cal G} = ({\cal N}, {\cal L})$ with link capacity constraints~\eqref{eqn:rout_cost_sum}-\eqref{eqn:rout_cost_non_neg}, and with VIP queue evolution~\eqref{eqn:queue_dyn}, is the set $\Lambda$ consisting of all  $(\lambda_n^k)_{k \in {\cal K}, n \in {\cal N}}$  such that there exist flow variables $(f_{ab}^k)_{k \in {\cal K}, (a,b) \in {\cal L}}$ and storage variables  $(\beta_{n,i,l})_{ n \in {\cal N}; i=1,\cdots, {K \choose l};\  l=0,\cdots, i_n \triangleq \lfloor L_n/D\rfloor}$
satisfying
\begin{align}
&f_{ab}^k\geq 0,  f_{nn}^k=0,   f_{src(k)n}^k=0, \  \forall a,b, n\in \mathcal N,  k\in \mathcal K\label{eqn:stability-region-f1}\\
&f_{ab}^k=0, \ \forall a,b\in \mathcal N, \ k\in \mathcal K,\  (a,b)\not \in \mathcal L^k\label{eqn:stability-region-f2}\\
&0\leq \beta_{n,i,l}\leq 1, \; i=1,\cdots, {K \choose l},\ l=0,\cdots, i_n, \ n\in \mathcal N\label{eqn:stability-region-beta}\\
&\frac{\lambda_n^k}{\theta_n^k}\leq  \sum_{b\in \mathcal N}f^k_{nb}-\frac{\sum_{a\in \mathcal N}f^k_{an}}{\theta_n^k}+ r_n\sum_{l=0}^{i_n}\sum_{i=1}^{{K \choose l}}\beta_{n,i,l}\mathbf 1[k\in \mathcal B_{n,i,l}], \nonumber \\
&\hspace{30mm} \forall n\in \mathcal N,\ k\in \mathcal K, n\neq src(k) \label{eq:sink} \\
&\sum_{k\in \mathcal K} f^k_{ab}\leq C_{ba}/D, \ \forall (a,b)\in \mathcal L\label{eqn:stability-region-capacity}\\
&\sum_{l=0}^{i_n}\sum_{i=1}^{{K \choose l}} \beta_{n,i,l}=1, \ \forall n\in \mathcal N\label{eqn:stability-region-cache}
\end{align}
Here,
$\mathcal B_{n,i,l}$ denotes the caching set consisting of the $i$-th combination of $l$ data objects out of $K$ data objects at node $n$, where $i=1,\cdots, {K \choose l}$, $l=0,\cdots, i_n \triangleq\lfloor L_n/D\rfloor$.
\label{thm:stability}
\end{Thm}
\begin{proof} Please refer to Appendix A. 
\end{proof}

To interpret Theorem~\ref{thm:stability}, note that the flow variable $f_{ab}^k$ represents the long-term VIP flow rate for data object $k$ over link $(a,b)$.  The
storage variable $\beta_{n,i,l}$ represents the long-term fraction of time that the set $\mathcal B_{n,i,l}$ (the $i$-th combination of $l$ data objects out of $K$ data objects) is cached at node $n$.  Inequality~\eqref{eq:sink} states that the  scaled-down exogenous VIP arrival rate for data object $k$ at node $n$ is upper bounded by the total long-term outgoing VIP flow rate minus the total  scaled-down endogenous long-term incoming VIP flow rate, plus the long-term VIP flow rate which is absorbed by all possible caching sets containing data object $k$ at node $n$, weighted by the fraction of time each caching set is used.  It is easy to see that the stability region in Theorem~\ref{thm:stability} (with VIP scaling) is a superset of the stability region  in Theorem~2 of \cite{VIPICN14} (without VIP scaling).  It is also clear that the stability region in Theorem~\ref{thm:stability} becomes larger when $\theta_n^k$ increases for all $k\in\mathcal K$ and $n\in\mathcal N$.    Note that when $\theta_n^k=1$ for all $k\in\mathcal K$ and $n\in\mathcal N$, the VIP stability region in Theorem~\ref{thm:stability} reduces to the VIP stability region in Theorem~2 of \cite{VIPICN14}.


\section{Throughput Optimal VIP Control}\label{sec:forwarding-caching-VIP}

First, we  present a  new    joint dynamic  forwarding and caching algorithm for VIPs in the virtual plane when $\boldsymbol \lambda\in \text{int}(\Lambda)$.
\begin{Alg}[Scaled Forwarding and Caching]  At the beginning of each slot  $t$, observe the VIP counts $\mathbf V(t)\triangleq (V^k_n(t))_{ n \in \mathcal N,k\in \mathcal K}$ and perform forwarding and caching in the virtual plane as follows.

\textbf{Forwarding}: For each data object $k \in {\cal K}$ and each link $(a,b)\in \mathcal
L^k$,  choose
\begin{align}
\mu^{k}_{ab}(t)
=&\begin{cases} C_{ba}/D,
&  W^*_{ab}(t)>0\  \text{and}\ k=k^*_{ab}(t)\\
0, &\text{otherwise}
\end{cases}\label{eqn:forwarding-VIP}
\end{align}
$W_{ab}^{k}(t) \triangleq  V^k_a(t) - \frac{V^k_b(t)}{\theta_b^k}$,   $W^*_{ab}(t)
\triangleq \left(W_{ab}^{k^*_{ab}(t)}(t)\right)^+$, and $k^*_{ab}(t) \triangleq \arg
\max_{k\in\{k: (a,b)\in \mathcal L^{k}\}} W_{ab}^{k}(t)$.
Here,
$W_{ab}^{k}(t)$ is the backpressure weight of object $k$ on link $(a,b)$ at slot $t$, and
$k^*_{ab}(t)$ is the data object which maximizes the backpressure weight on link $(a,b)$ at time $t$.

\textbf{Caching}:  At each node $n \in \mathcal N$, choose $(s^k_n(t))_{n\in\mathcal N, k\in\mathcal K}$ to
\begin{equation}
\max \sum_{k\in \mathcal K} V^k_n(t) s^k_n \quad
\text{s.t.} \ \sum_{k\in \mathcal K} s^k_n\leq L_n/D.
\label{eqn:knapsack-VIP}
\end{equation}
 Here, $V^k_n(t)$ serves as the  caching weight of object $k$ at node $n$.

Based on the forwarding and caching in \eqref{eqn:forwarding-VIP} and \eqref{eqn:knapsack-VIP}, the VIP count is updated according to \eqref{eqn:queue_dyn}.
\label{Alg:VIP}
\end{Alg}


 Note that when $\theta_n^k=1$ for all $k\in\mathcal K$ and $n\in\mathcal N$, Algorithm~\ref{Alg:VIP} reduces to Algorithm~1 in \cite{VIPICN14}. The computation complexity of Algorithm~\ref{Alg:VIP} has the same order as that of  Algorithm~1 in \cite{VIPICN14}. Later, in Section~\ref{sec:simulations}, we shall see that Algorithm~\ref{Alg:VIP} yields superior delay performance to  Algorithm~1 in \cite{VIPICN14}.
The forwarding part in Algorithm~\ref{Alg:VIP} is different from that in Algorithm~1 in \cite{VIPICN14}.
At each slot $t$ and for each link $(a,b)$, the scaled backpressure-based forwarding algorithm allocates the entire normalized ``reverse" link capacity $C_{ba}/D$ to transmit the VIPs for the data object $k^*_{ab}(t)$ which maximizes the backpressure $W_{ab}^{k}(t) $.   The forwarding algorithm captures to some extent the interest suppression at receiving node $b$.  
The caching part in Algorithm~\ref{Alg:VIP} is the same as that in Algorithm~1 in \cite{VIPICN14}.   The  max-weight  caching algorithm  implements the optimal solution to the max-weight knapsack problem in \eqref{eqn:knapsack-VIP}, i.e., allocate cache space at node $n$ to the $\lfloor L_n/D\rfloor$ objects with the largest  caching weights $V^k_n(t)$.
The scaled forwarding and caching algorithm maximally balances out the scaled VIP counts by joint forwarding and caching, in order to prevent congestion building up in any part of the network, thereby reducing delay.
 Similar to Algorithm~1 in \cite{VIPICN14}, Algorithm~\ref{Alg:VIP} can be implemented in a distributed manner.

We now show that   Algorithm~\ref{Alg:VIP} adaptively stabilizes all VIP queues for any ${\boldsymbol \lambda} \in {\rm
int}({\Lambda})$, without knowing ${\boldsymbol \lambda}$. Thus, Algorithm~\ref{Alg:VIP} is {\em throughput optimal}, in the sense of adaptively maximizing the VIP throughput, and therefore the user demand rate satisfied by the network.

\begin{Thm} [Throughput Optimality] If there exists  $\boldsymbol \epsilon=(\epsilon_n^k)_{n \in {\cal N}, k \in {\cal K}} \succ \mathbf 0$
such
that $\boldsymbol \lambda+\boldsymbol \epsilon \in \Lambda$,
then the network of VIP queues under
Algorithm \ref{Alg:VIP} satisfies
\begin{align}
\limsup_{t\to\infty}\frac{1}{t}\sum_{\tau=1}^{t}\sum_{n\in \mathcal N,k\in \mathcal K} \mathbb
E[V^k_n(\tau)]\leq \frac{N B}{\epsilon}\label{eqn:enhanced-DBP-VB}
\end{align}
where
$ B\triangleq \frac{1}{2N}\sum_{n\in \mathcal N}\big((\mu^{out}_{n,
\max})^2+\frac{A_{n,\max}+\mu^{in}_{n,\max}}{\theta_{n,min}}+r_{n,\max})^2
+2\mu^{out}_{n,
\max}r_{n,\max}\big)$,
$\epsilon\triangleq \min_{ n\in \mathcal
N,k\in \mathcal K} \epsilon_n^k$,
$\mu^{in}_{n, \max}\triangleq \sum_{a\in \mathcal N} C_{an}/D $, $\mu^{out}_{n, \max}\triangleq
\sum_{b\in \mathcal N} C_{nb}/D$, $\theta_{n,\min}  \triangleq  \min_{k\in \mathcal K} \theta_n^k$, $
A_{n,\max} \triangleq \sum_{k\in \mathcal K}
A^k_{n,\max}$, and $r_{n,\max}= K r_n$.\label{Thm:thpt-opt}
\end{Thm}
\begin{proof} Please refer to Appendix B.
\end{proof}
 The upper bound \eqref{eqn:enhanced-DBP-VB} on the average total number of VIPs is smaller than the upper bound in Theorem~1 of  \cite{VIPICN14}.  This is because during transmission within the modified virtual plane,  the number of VIPs  for an object is scaled down.  In addition, the upper bound  in \eqref{eqn:enhanced-DBP-VB} decreases  when $\theta_n^k$ increases for all $k\in\mathcal K$ and $n\in\mathcal N$.
Note that when $\theta_n^k=1$ for all $k\in\mathcal K$ and $n\in\mathcal N$, Theorem~\ref{Thm:thpt-opt} reduces to Theorem~1 in \cite{VIPICN14}.

\begin{figure*}[t]
\begin{center}
\includegraphics[width=17.4cm, height=4cm]{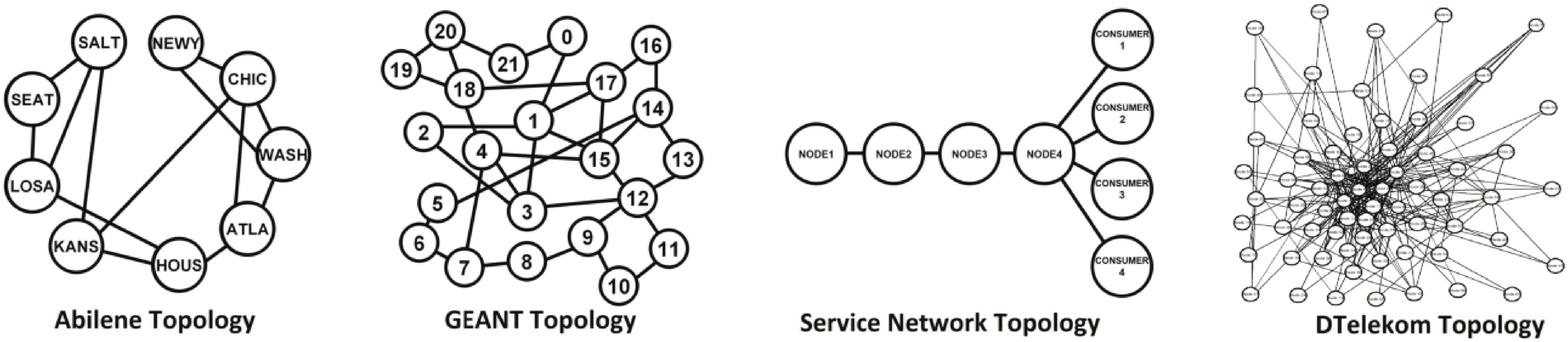}
\caption{Network topologies \cite{VIPICN14}.}
\end{center}
\label{Fig:topo}
\end{figure*}
  \section{Experimental Evaluation}
\label{sec:simulations}
\begin{figure}[t]
\begin{center}
\vspace{-\baselineskip}
\centering
\includegraphics[height=6.8cm,width=8.5cm]{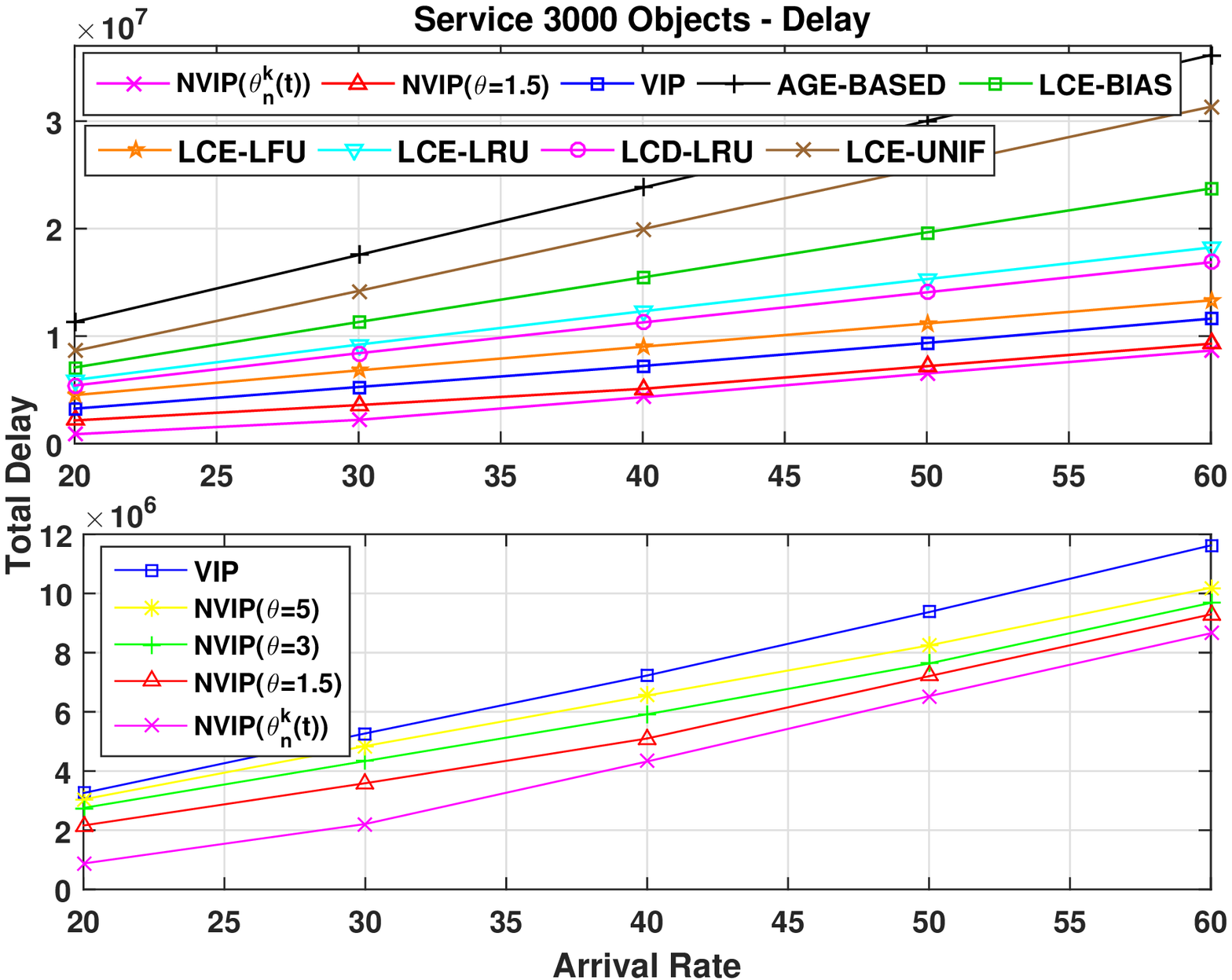}
\caption{Average delay for Service Topology.}\label{fig:service}

\vspace{10pt}
\centering
\includegraphics[height=6.8cm,width=8.5cm]{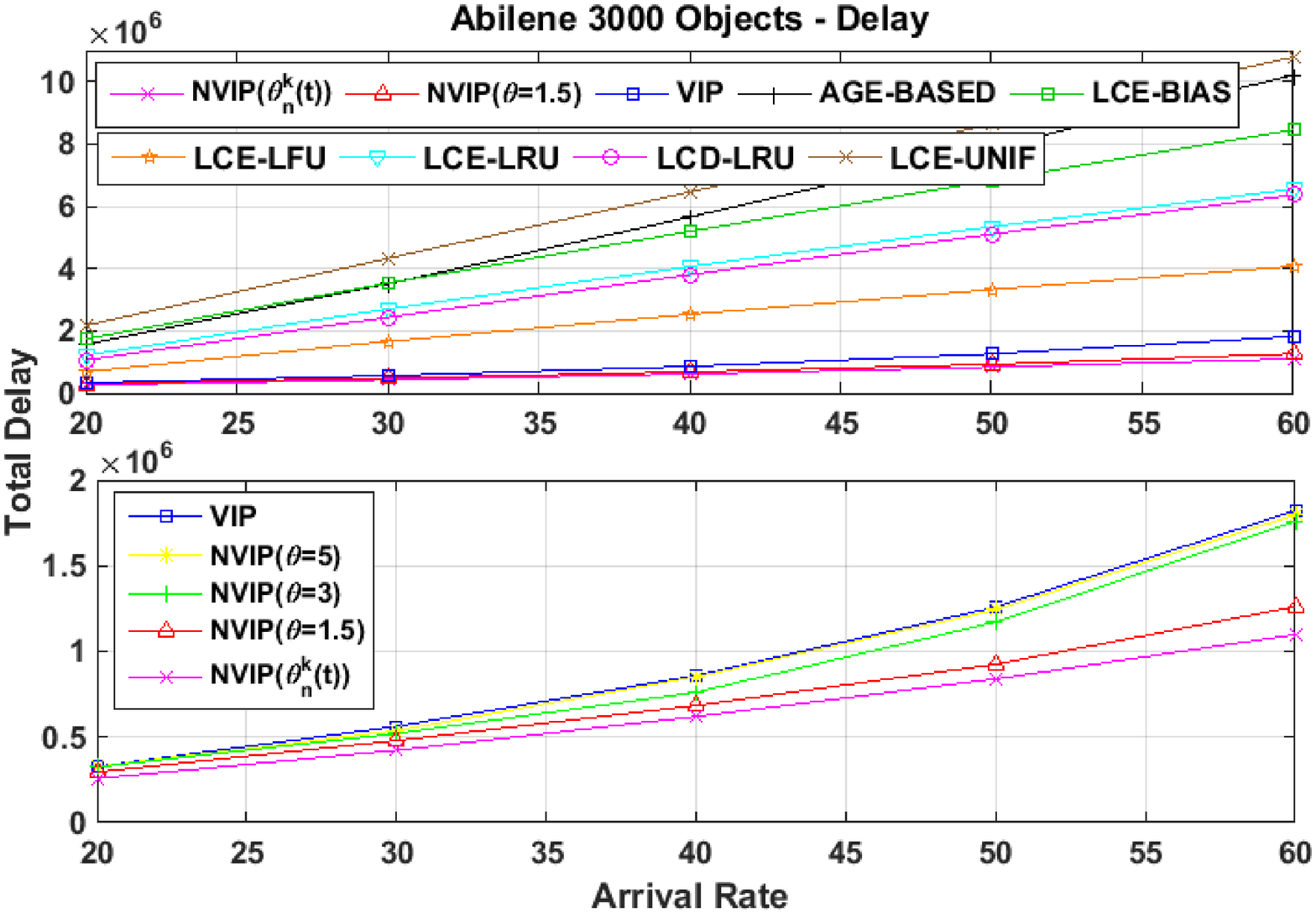}
\caption{Average delay for Abilene Topology.}\label{fig:Abilene}
\end{center}
\end{figure}

%
%
Using the Scaled VIP algorithm in Algorithm~\ref{Alg:VIP}, we can develop a corresponding  stable caching VIP algorithm  for handling Interest Packets and Data Packets in the actual plane using the mapping in \cite{VIPICN14} (please see therein for details). We now compare the delay  performance of  the   new VIP algorithm for the actual plane  resulting from Algorithm~\ref{Alg:VIP}, denoted by NVIP,  with the performance of the VIP algorithm for the actual plane resulting from Algorithm~1 in \cite{VIPICN14}, denoted by VIP, and with the performance of six other baseline algorithms. These baseline algorithms use popular caching algorithms (LFU, LCE-UNIF, LCE-LRU, LCD-LRU, and LCE-BIAS)  in conjunction with shortest path forwarding and a potential-based forwarding algorithm. The detailed descriptions of these baseline algorithms can be found  in \cite{VIPICN14}.

\begin{figure}[h]
\begin{center}
\vspace{-\baselineskip}
\centering
\includegraphics[height=6.8cm,width=8.5cm]{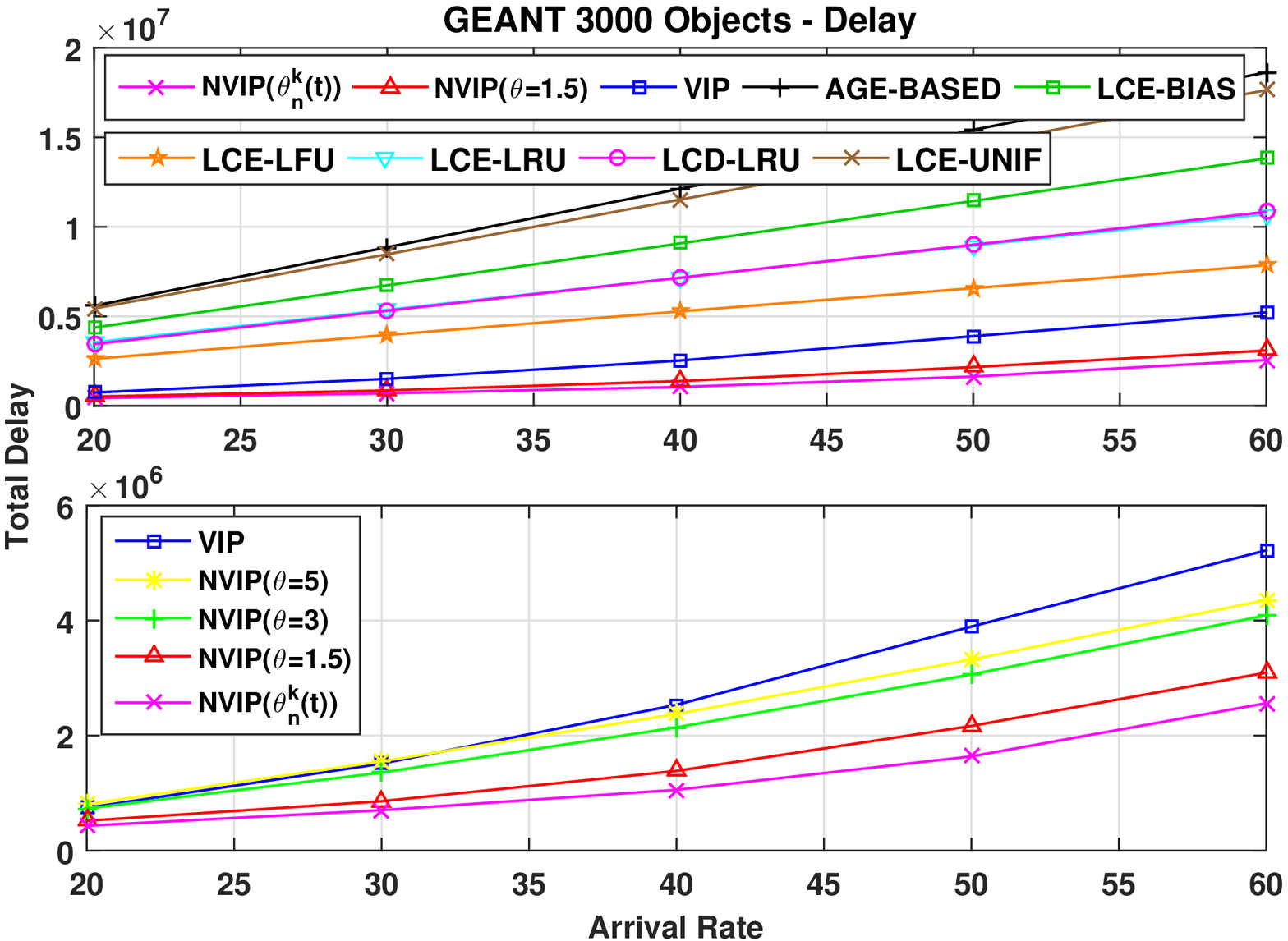}
\caption{Average delay for GEANT Topology.}\label{fig:GEANT}

\vspace{10pt}
\centering
\includegraphics[height=6.8cm,width=8.5cm]{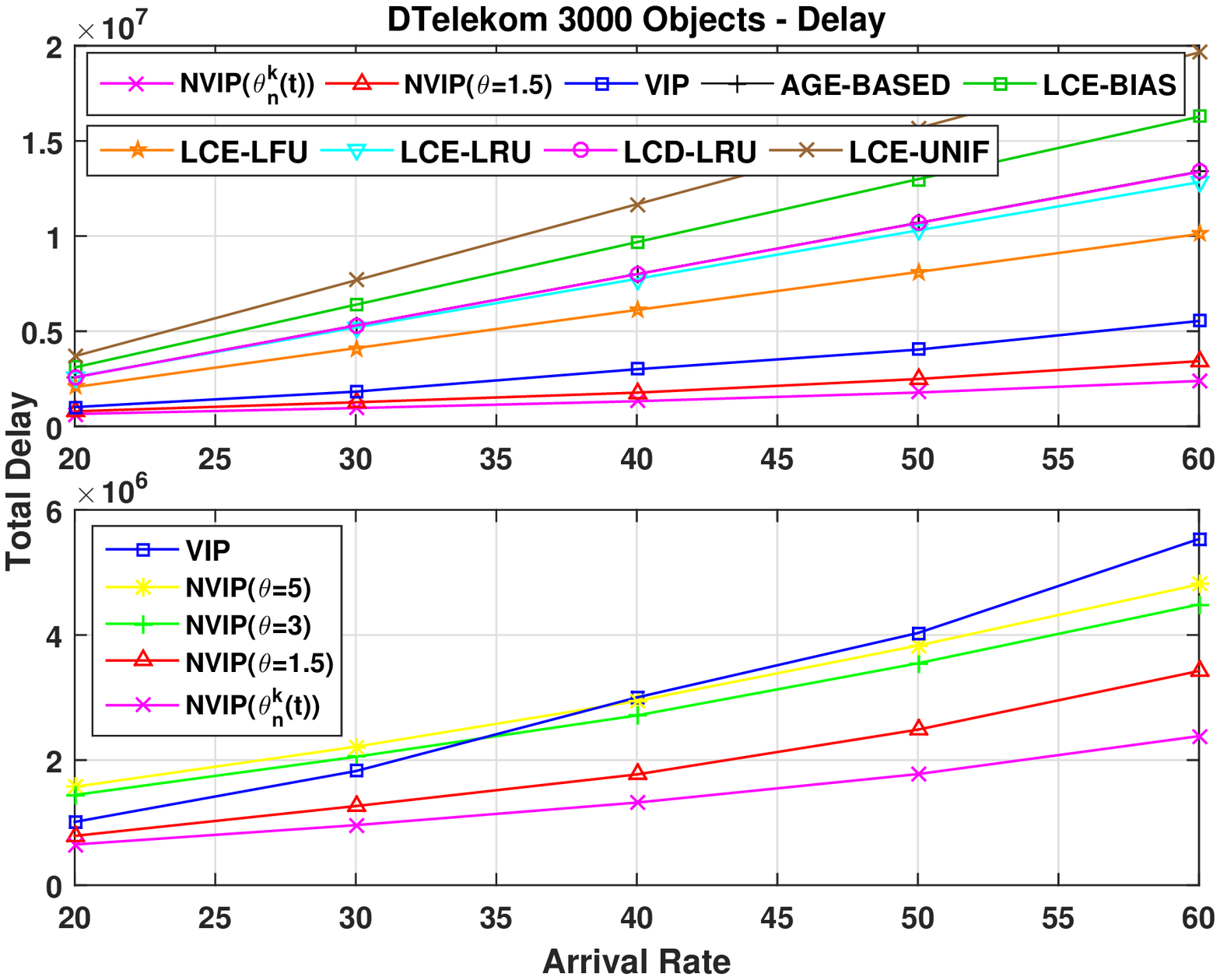}
\caption{\small{Average delay for DTelekom Topology.}}\label{fig:DTelekom}
\end{center}
\end{figure}
We consider two ways of choosing $\theta_n^k$. In the first way, for simplicity, we choose $\theta_n^k=\theta$  for all $k\in\mathcal K$ and $n\in\mathcal N$. In the second way, we choose $\theta_n^k$ to be a moving average of $A^k_n(t)+
\sum_{a\in \mathcal N}\mu^{k}_{an}(t)$ for all $k\in\mathcal K$ and $n\in\mathcal N$. In particular, we consider an exponential moving average (EMA), and  update  $\theta_n^k(t)$ according to  $\theta_n^k(t) = (1 - \beta)\theta_n^k(t-1)+\beta\left(A_n^k(t)+\sum_{a\in \mathcal N}\mu_{an}^k(t)\right)$
at each time $t$, where $\beta $ is set to be 0.125 in the simulation.


Experimental scenarios are carried on four network topologies: the Service Topology, the Abilene Topology, the GEANT Topology and the DTelekom Topology, as shown in Fig.~\ref{Fig:topo}. In the Service Topology, NODE 1 is the content source for all objects, requests can be generated only by the CONSUMER nodes, and  the cache size  at each node is 5 $GB$. The cache size at each node is 5 $GB$ in  the Abilene Topology and is 2 $GB$ in the  GEANT Topology and the DTelekom Topology. In the Abilene Topology, the GEANT Topology and the DTelekom Topology, object requests can be generated by any node, and the content source for each data object is independently and uniformly distributed among all nodes.  At each node requesting data, object requests arrive according to a Poisson process  with an overall rate $\lambda$ (in requests/node/slot). Each arriving request requests data object $k$ (independently) with probability $p_k$, where $\{p_k\}$ follows a (normalized) Zipf distribution with parameter 0.75.
We choose $K=3000$, $C_{ab}=500$ Mb/slot, and $D=5$ MB.  The Interest Packet size is 125B, and the Data Packet size is 50 KB.  Each simulation generates requests  for  $10^4$ time slots.  Each curve   is obtained by averaging over 10 simulation runs.
The delay  for an Interest Packet request is the difference  (in time slots) between the fulfillment
time (i.e., time of arrival of the requested Data Packet) and the creation time
of the Interest Packet request. We use the total delay for all the Interest Packets generated over $10^4$ time slots as the delay measure.


%
%
%
%
%
%
%
%
%
%

Fig.~\ref{fig:service}--Fig.~\ref{fig:DTelekom} illustrate the delay performance for the four topologies. From these figures, we can observe that the VIP algorithm in \cite{VIPICN14}  and the proposed new VIP algorithms with different scaling constants $\theta$ achieve much better delay performance than the six baseline schemes. In addition,
 the proposed new VIP algorithms with $\theta=1$ and EMA $\theta_n^k(t)$  significantly improve the  delay performance of the VIP algorithm  (e.g.,
 $58\%$ at $\lambda=30$ in Service, $40\%$ at   $\lambda=60$ in Abilene, $58\%$ at   $\lambda=40$ in GEANT and $57\%$ at $\lambda=60$ in DTelekom with EMA $\theta_n^k(t)$), indicating the advantage of the proposed VIP scaling in reflecting the interest suppression in the actual network. Note that a large $\theta$ may overemphasize interest suppression and underestimate measured demand, leading to performance degradation.

\section{Conclusion}


In this paper, we aimed to further improve the  performance of the existing VIP algorithms by using a modified virtual plane where VIP counts are appropriately scaled to reflect interest suppression effects.
We characterized the scaled VIP stability region in the virtual plane,
developed new distributed forwarding and caching algorithms,
and proved the throughput optimality of the proposed algorithms in the virtual plane.
 Numerical experiments demonstrate significantly enhanced performance relative to the existing VIP algorithm, as well as a number of other baseline algorithms.

\section*{Appendix A: Proof of Theorem~\ref{thm:stability}}

The proof of Theorem \ref{thm:stability} involves showing that $\boldsymbol\lambda\in\Lambda$ is necessary for stability and that $\boldsymbol \lambda\in\text{int}(\Lambda)$ is sufficient for stability.
First, we show $\boldsymbol \lambda\in\Lambda$ is necessary for stability.  Suppose the network under arrival rate  $\boldsymbol \lambda$ is stabilizable by some feasible forwarding and caching policy. Let $F_{ab}^k(t)$ denote the number of VIPs for object $k$ transmitted over link $(a,b)$ during slot $t$, satisfying
\begin{align}
&F_{ab}^k(t)\geq 0, \ F_{nn}^k(t)=0, \ F_{src(k)n}^k(t)=0, \quad\nonumber\\
&\hspace{40mm} \forall a,b, n\in \mathcal N, k\in \mathcal K\label{eqn:proof-stability-region-F1}\\
&F_{ab}^k(t)=0, \quad \forall a,b\in \mathcal N, k\in \mathcal K, (a,b)\not \in \mathcal L^k\label{eqn:proof-stability-region-F2}\\
&\sum_{k\in \mathcal K} F_{ab}^k(t)\leq C_{ba}/z, \quad \forall (a,b)\in \mathcal L\label{eqn:proof-capacity-const}
\end{align}
For any slot $\tilde t$, we can define $f_{ab}^k=\sum_{\tau=1}^{\tilde t}F_{ab}^k(\tau)/\tilde t$.
Thus, by \eqref{eqn:proof-stability-region-F1}, \eqref{eqn:proof-stability-region-F2}, and \eqref{eqn:proof-capacity-const}, we can prove \eqref{eqn:stability-region-f1}, \eqref{eqn:stability-region-f2}, and \eqref{eqn:stability-region-capacity}, separately.
Let $S_n^k(t)$ denote the caching state of object $k$ at node $n$ during slot $t$, which satisfies
\begin{align}
&S_n^k(t)\in \{0,1\}, \quad \forall n\in \mathcal N, k\in \mathcal K\label{eqn:stability-region-S}
\end{align}
Define \footnote{Note that $\mathcal T_{n,i,l}\cap \mathcal T_{n,j,m}=\emptyset$ for all $(i,l)\neq (j,m)$ for all $n\in \mathcal N$.}
\begin{align}
\mathcal T_{n,i,l}=\bigg\{\tau\in \{1,\cdots, \tilde t\}:&S_n^k(\tau)=1 \ \forall k\in \mathcal B_{n,i,l}, \nonumber\\ &S_n^k(\tau)=0  \ \forall k\not\in \mathcal B_{n,i,l} \bigg\}\nonumber
\end{align}
for $i=1,\cdots, {K \choose l}$ and $l=0,\cdots, i_n$. Define $\beta_{n,i,l}=T_{n,i,l}/\tilde t$, where  $T_{n,i,l}=|\mathcal T_{n,i,l}|$. Thus, we can prove \eqref{eqn:stability-region-beta} and  \eqref{eqn:stability-region-cache}. It remains to prove \eqref{eq:sink}.
By Lemma 1 of \cite{Neely-Modiano-Rohrs:2005}, network stability implies there exists a finite $M$ such that $V_n^k(t)\leq M$ for all $n\in \mathcal N$ and $k\in \mathcal K$ holds infinitely often. Given an arbitrarily small value $\epsilon >0$, there exists a slot $\tilde t$ such that
\begin{align}
V_n^k(\tilde t)\leq M,\quad 
\frac{M}{\tilde t}\leq \epsilon, \quad
\left|\frac{\sum_{\tau=1}^{\tilde t}A_n^k(\tau)}{\tilde t}-\lambda_n^k\right|\leq \epsilon\label{eqn:proof-queue-bound}
\end{align}
In addition, since for all slot $t$, the queue length is equal to the difference between the total VIPs that have arrived and departed as well as drained, assuming $V_n^k(1)=0$, we have
\begin{align}
&\sum_{\tau=1}^t \dfrac{A_{n}^k(\tau)}{\theta_n^k}-V_n^k(t)\nonumber\\
\leq&\sum_{\tau=1}^t\sum_{b\in \mathcal N} F_{nb}^k(\tau)-\sum_{\tau=1}^t\sum_{a\in \mathcal N} F_{an}^k(\tau)/\theta_n^k+r_n\sum_{\tau=1}^tS_n^k(\tau)\label{eqn:proof-queue-const}
\end{align}
Thus,
by \eqref{eqn:proof-queue-bound} and \eqref{eqn:proof-queue-const},  we have
\begin{align}
&\lambda_n^k-\epsilon\leq \frac{1}{\tilde t}\sum_{\tau=1}^{\tilde t} A_{n}^k(\tau)\nonumber\\
\leq&\theta_n^k\Big(\frac{1}{\tilde t}V_n^k(\tilde t)+\frac{1}{\tilde t}\sum_{\tau=1}^{\tilde t} \sum_{b\in \mathcal N} F_{nb}^k(\tau)-\frac{1}{\tilde t}\sum_{\tau=1}^{\tilde t} \sum_{a\in \mathcal N} \dfrac{F_{an}^k(\tau)}{\theta_n^k}\nonumber\\
&+r_n\frac{1}{\tilde t}\sum_{\tau=1}^{\tilde t} S_n^k(\tau)\Big)\label{eqn:proof-stability-region-lambda}
\end{align}
Since $ \sum_{\tau=1}^{\tilde t} S_n^k(\tau)=
\sum_{l=0}^{i_n}\sum_{i=1}^{{K \choose l}} T_{n,i,l}\mathbf 1[k\in \mathcal B_{n,i,l}]$, by
\eqref{eqn:proof-stability-region-lambda}, we have
$$\lambda_n^k\leq \theta_n^k\sum_{b\in \mathcal N}f^k_{nb}-\sum_{a\in \mathcal N}f^k_{an}+\theta_n^k r_n\sum_{l=0}^{i_n}\sum_{i=1}^{{K \choose l}} \beta_{n,i,l}\mathbf 1[k\in \mathcal B_{n,i,l}]+\theta_n^k\epsilon + \epsilon.$$
By letting $\epsilon\to 0$, we can prove \eqref{eq:sink}.

Next, we show $\boldsymbol \lambda\in\text{int}(\Lambda)$ is sufficient for stability. $\boldsymbol \lambda\in\text{int}(\Lambda)$ implies that there exists $\boldsymbol \epsilon=\left(\epsilon_n^k\right)$, where $\epsilon_n^k>0$, such that $\boldsymbol \lambda+\boldsymbol \epsilon\in \Lambda$. Let $\left(f_{ab}^k\right)$ and $\left(\boldsymbol \beta_n\right)$ denote the  flow variables  and storage variables associated with arrival rates $\boldsymbol \lambda+\boldsymbol \epsilon$. Thus, \eqref{eqn:stability-region-f1}, \eqref{eqn:stability-region-f2}, \eqref{eqn:stability-region-beta}, \eqref{eqn:stability-region-capacity}, \eqref{eqn:stability-region-cache}, and
\begin{align}
&\lambda_n^k+\epsilon_n^k\nonumber\\
\leq& \theta_n^k\sum_{b\in \mathcal N}f^k_{nb}-\sum_{a\in \mathcal N}f^k_{an}+\theta_n^k r_n\sum_{l=0}^{i_n}\sum_{i=1}^{{K \choose l}} \beta_{n,i,l}\mathbf 1[k\in \mathcal B_{n,i,l}],\nonumber\\
&\hspace{30mm}\forall n\in \mathcal N, k\in \mathcal K, n\neq src(k)\label{eqn:stability-region-lambda-epsilon}
\end{align}
hold. We now construct the randomized forwarding policy. For every link $(a,b)$ such that $\sum_{k\in \mathcal K}f_{ab}^k>0$, transmit the VIPs of the single object $\tilde k_{ab}$, where  $\tilde k_{ab}$ is chosen randomly to be $k$ with probability $f_{ab}^k/\sum_{k\in \mathcal K}f_{ab}^k$. Then, the number of VIPs that can be transmitted in slot $t$ is as follows:
\begin{align}
\tilde \mu_{ab}^k(t)=
\begin{cases}
\sum_{k\in \mathcal K}f_{ab}^k, & \text{if  $k=\tilde k_{ab}$} \\
0, & \text{otherwise}
\end{cases}
\end{align}
Null bits are delivered if there are not enough bits in a queue. For every link $(a,b)$ such that $\sum_{k\in \mathcal K}f_{ab}^k=0$, choose $\tilde \mu_{ab}^k(t)=0$ for all $k\in\mathcal K$. Thus, we have
\begin{align}
\mathbb E\left[\tilde \mu_{ab}^k(t)\right]=f_{ab}^k\label{eqn:rand-mu-ave}
\end{align}
Next, we construct the randomized caching policy. For every node $n$, cache the single combination $\tilde{\mathcal B}_n$, where $\tilde{\mathcal B}_n$ is chosen randomly to be $\mathcal B_{n,i,l}$ with probability $\beta_{n,i,l}/\sum_{l=0}^{i_n}\sum_{i=1}^{{K \choose l}} \beta_{n,i,l}=\beta_{n,i,l},$ as $\sum_{l=0}^{i_n}\sum_{i=1}^{{K \choose l}} \beta_{n,i,l}=1$ by \eqref{eqn:stability-region-cache}. Then, the caching state in slot $t$ is as follows:
\begin{align}
\tilde s_n^k(t)=
\begin{cases}
1, & \text{if  $k\in \tilde{\mathcal B}_n$} \\
0, & \text{otherwise}
\end{cases}
\end{align}
Thus, we have
\begin{align}
\mathbb E\left[\tilde s_n^k(t)\right]=\sum_{l=0}^{i_n}\sum_{i=1}^{{K \choose l}} \beta_{n,i,l}\mathbf 1[k\in \mathcal B_{n,i,l}] \label{eqn:rand-s-ave}
\end{align}
Therefore, by \eqref{eqn:rand-mu-ave}, \eqref{eqn:rand-s-ave} and \eqref{eqn:stability-region-lambda-epsilon}, we have
\begin{align}
&\mathbb E\left[\left(\theta_n^k\sum_{b\in \mathcal N}\tilde{\mu}^k_{nb}(t)-\sum_{a\in \mathcal N}\tilde{\mu}^k_{an}(t)+\theta_n^k r_n^{(k)}\tilde{s}_n^k(t)\right)\right]\nonumber\\
=&\theta_n^k\sum_{b\in \mathcal N}f^k_{nb}-\sum_{a\in \mathcal N}f^k_{an}+\theta_n^k r_n\sum_{l=0}^{i_n}\sum_{i=1}^{{K \choose l}} \beta_{n,i,l}\mathbf 1[k\in \mathcal B_{n,i,l}]\nonumber\\
\geq &\lambda_n^k+\epsilon_n^k\label{eqn:rand-policy-ineq}
\end{align}
In other words, the arrival rate is less than the service rate.
Thus, by Loynes' theorem, we can show that the network is stable.

\section*{Appendix B: Proof of Theorem~\ref{Thm:thpt-opt}}

Define the quadratic Lyapunov function $\mathcal L(\mathbf V)\triangleq\\
\sum_{n\in \mathcal N,k\in \mathcal K}(V^k_n)^2$. The Lyapunov drift at slot $t$ is given by
$\Delta (\mathbf V(t))\triangleq \mathbb E[\mathcal L\big(\mathbf
V(t+1)\big)-\mathcal L\big(\mathbf V(t)\big)|\mathbf V(t)]$. First, we calculate $\Delta (\mathbf V(t))$.
Taking square on both sides of \eqref{eqn:queue_dyn}, we have
\begin{align}
&\left(V^k_n(t+1)\right)^2\nonumber\\
\leq&\Bigg(\Bigg(
\left(V^k_n(t)-\sum_{b\in \mathcal N}\mu^k_{nb}(t)\right)^+ +\dfrac{A^k_n(t)}{\theta_n^k}
\nonumber\\
&+\sum_{a\in \mathcal N}\mu^k_{an}(t)/\theta_n^k-r_ns_n^k(t)\Bigg)^+\Bigg)^2\nonumber\\
\leq&\Bigg(
\left(V^k_n(t)-\sum_{b\in \mathcal N}\mu^k_{nb}(t)\right)^+ +\dfrac{A^k_n(t)}{\theta_n^k}\nonumber\\
&+\sum_{a\in \mathcal N}\mu^k_{an}(t)/\theta_n^k-r_ns_n^k(t)\Bigg)^2\nonumber\\
\leq&\left(V^k_n(t)-\sum_{b\in \mathcal N}\mu^k_{nb}(t)\right)^2+2
\left(V^k_n(t)-\sum_{b\in \mathcal N}\mu^k_{nb}(t)\right)^+\nonumber\\
&\times\left(\dfrac{A^k_n(t)}{\theta_n^k}
+\sum_{a\in \mathcal N}\dfrac{\mu^k_{an}(t)}{\theta_n^k}-r_ns_n^k(t)\right)\nonumber\\
&+\left(\dfrac{A^k_n(t)}{\theta_n^k}
+\sum_{a\in \mathcal N}\dfrac{\mu^k_{an}(t)}{\theta_n^k}-r_ns_n^k(t)\right)^2
\nonumber\\
=&\left(V^k_n(t)\right)^2+\left(\sum_{b\in \mathcal N}\mu^k_{nb}(t)\right)^2-2 V^k_n(t)\sum_{b\in \mathcal N}\mu^k_{nb}(t) \nonumber\\
&+\left(\dfrac{A^k_n(t)}{\theta_n^k}+\sum_{a\in \mathcal N}\dfrac{\mu^k_{an}(t)}{\theta_n^k}-r_n^ks_n^k(t)\right)^2\nonumber\\
&+2\left(V^k_n(t)-\sum_{b\in \mathcal N}\mu^k_{nb}(t)\right)^+\left(\dfrac{A^k_n(t)}{\theta_n^k}
+\sum_{a\in \mathcal N}\dfrac{\mu^k_{an}(t)}{\theta_n^k}\right)\nonumber\\
&-2\left(V^k_n(t)-\sum_{b\in \mathcal N}\mu^k_{nb}(t)\right)^+r_n^ks_n^k(t)\nonumber\\
\leq&\left(V^k_n(t)\right)^2+\left(\sum_{b\in \mathcal N}\mu^k_{nb}(t)\right)^2-2 V^k_n(t)\sum_{b\in \mathcal N}\mu^k_{nb}(t) \nonumber\\
&+\left(\dfrac{A^k_n(t)}{\theta_n^k}+\sum_{a\in \mathcal N}\dfrac{\mu^k_{an}(t)}{\theta_n^k}+r_n^ks_n^k(t)\right)^2\nonumber\\
&+2V^k_n(t)\left(\dfrac{A^k_n(t)}{\theta_n^k}
+\sum_{a\in \mathcal N}\dfrac{\mu^k_{an}(t)}{\theta_n^k}\right)\nonumber\\
&-2\left(V^k_n(t)-\sum_{b\in \mathcal N}\mu^k_{nb}(t)\right)r_n^ks_n^k(t)\nonumber\\
\leq&\left(V^k_n(t)\right)^2+\left(\sum_{b\in \mathcal N}\mu^k_{nb}(t)\right)^2+2\sum_{b\in \mathcal N}\mu^k_{nb}(t)r_ns_n^k(t)\nonumber\\
&+\left(\dfrac{A^k_n(t)}{\theta_n^k}+\sum_{a\in \mathcal N}\dfrac{\mu^k_{an}(t)}{\theta_n^k}+r_ns_n^k(t)\right)^2\nonumber\\
&+2V^k_n(t)\dfrac{A^k_n(t)}{\theta_n^k}-2V^k_n(t)\left(\sum_{b\in \mathcal N}\mu^k_{nb}(t)-\sum_{a\in \mathcal N}\dfrac{\mu^k_{an}(t)}{\theta_n^k}\right)\nonumber\\
&-2V^k_n(t)r_ns_n^k(t)\nonumber
\end{align}
Summing over all $n,k$, we have
\begin{align}
&\mathcal L\left(\mathbf V(t+1)\right)-\mathcal L\left(\mathbf V(t)\right)\nonumber\\
\stackrel{(a)}{\leq}&2N B+2\sum_{n\in \mathcal N,k\in \mathcal K}V^k_n(t)\dfrac{A^k_n(t)}{\theta_n^k}\nonumber\\
&-2\sum_{(a,b)\in \mathcal L}\sum_{k\in \mathcal K}\mu_{ab}^k(t)(V_a^k(t)-\dfrac{V_b^k(t)}{\theta_b^k})\nonumber\\
&-2\sum_{n\in \mathcal N, k\in \mathcal K}V_n^k(t)r_n s_n^k(t)
\label{eqn:proof-deltaL}
\end{align}
where (a) is due to the following:
\begin{align}
&\sum_{k\in \mathcal K}\left(\sum_{b\in \mathcal N}\mu^k_{nb}(t)\right)^2\leq \left(\sum_{k\in \mathcal K}\sum_{b\in \mathcal N}\mu^k_{nb}(t)\right)^2\leq\left(\mu^{out}_{n, \max}\right)^2,\nonumber\\
& \sum_{k\in \mathcal K}\left(\dfrac{A^k_n(t)}{\theta_n^k}+\sum_{a\in \mathcal N}\dfrac{\mu^k_{an}(t)}{\theta_n^k}+r_ns_n^k(t)\right)^2\nonumber\\
&\leq \left(\sum_{k\in \mathcal K}\left(\dfrac{A^k_n(t)}{\theta_n^k}+\sum_{a\in \mathcal N}\dfrac{\mu^k_{an}(t)}{\theta_n^k}+r_ns_n^k(t)\right)\right)^2\nonumber\\
 &\leq (A_{n,\max}/\theta_{n,\min}+\mu^{in}_{n,\max}/\theta_{n,\min}+r_{n,\max})^2, \nonumber\\
&\sum_{k\in \mathcal K}\sum_{b\in \mathcal N}\mu^k_{nb}(t)r_ns_n^k(t)\nonumber\\
&\leq \left(\sum_{k\in \mathcal K}\sum_{b\in \mathcal N}\mu^k_{nb}(t)\right)\left(\sum_{k\in \mathcal K}r_ns_n^k(t)\right)\leq \mu^{out}_{n,
\max}r_{n,\max},\nonumber\\
&\sum_{n\in \mathcal N,k\in \mathcal K}V^k_n(t)\left(\sum_{b\in\mathcal N}\mu^k_{nb}(t)-\dfrac{1}{\theta_n^k}\sum_{a\in \mathcal N}\mu^k_{an}(t)\right)\nonumber\\
&=\sum_{(a,b)\in \mathcal L}\sum_{k\in \mathcal K}
\mu^k_{ab}(t)\big(V^k_a(t)-V^k_b(t)/\theta_b^k\big).\nonumber
\end{align}
Taking conditional
expectations on both sides of \eqref{eqn:proof-deltaL}, we have
\begin{align}
&\Delta (\mathbf
V(t))\nonumber\\
\leq&2N B+2\sum_{n\in \mathcal N,k\in \mathcal K}V^k_n(t)\dfrac{\lambda^k_n}{\theta_n^k}\nonumber\\
&-2\mathbb
E\left[\sum_{(a,b)\in \mathcal L}\sum_{k\in \mathcal K}
\mu^k_{ab}(t)\left(V^k_a(t)-\dfrac{V^{k}_b(t)}{\theta_b^k}\right)|\mathbf V(t)\right]\nonumber\\
&-2\mathbb E\left[\sum_{n\in \mathcal N,k\in \mathcal K}V^k_n(t)r_n s_n^k(t)|\mathbf V(t)\right]\nonumber\\
\stackrel{(b)}{\leq}&2N B+2\sum_{n\in \mathcal N,k\in \mathcal K}V^k_n(t)\dfrac{\lambda^k_n}{\theta_n^k}\nonumber\\
&-2\mathbb
E\left[\sum_{(a,b)\in \mathcal L}\sum_{k\in \mathcal K}
\tilde{\mu}^k_{ab}(t)\left(V^k_a(t)-\dfrac{V^k_b(t)}{\theta_b^k}\right)|\mathbf V(t)\right]\nonumber\\
&-2\mathbb E\left[\sum_{n\in \mathcal N,k\in \mathcal K}V^k_n(t)r_n\tilde{s}_n^k(t)|\mathbf V(t)\right]\nonumber\\
=&2N B+2\sum_{n\in \mathcal N,k\in \mathcal K}V^k_n(t)\dfrac{\lambda^k_n}{\theta_n^k}-2\sum_{n\in \mathcal N,k\in \mathcal K}V^k_n(t)\nonumber\\
&\times\mathbb E\left[\left(\sum_{b\in \mathcal N}\tilde{\mu}^k_{nb}(t)-\sum_{a\in \mathcal N}\tilde{\mu}^k_{an}(t)/\theta_n^k+r_n\tilde{s}_n^k(t)\right)|\mathbf
V(t)\right]\label{eqn:proof_ineq0}
\end{align}
where (b) is due to the fact that Algorithm \ref{Alg:VIP} minimizes the
R.H.S. of  (b) over all feasible $\tilde{\mu}^k_{ab}(t)$ and $\tilde{s}_n^k(t)$.\footnote{ Note that $\mu^k_{ab}(t)$ and $s_n^k(t)$ denote the actions of Algorithm \ref{Alg:VIP}.}  Since $\boldsymbol
\lambda+\boldsymbol \epsilon  \in \Lambda$, according to the proof of Theorem \ref{thm:stability}, there exists a stationary randomized forwarding and caching policy that makes decisions
independent of $\mathbf V(t)$ such that
\begin{align}
&\mathbb E\left[\left(\theta_n^k\sum_{b\in \mathcal N}\tilde{\mu}^k_{nb}(t)-\sum_{a\in \mathcal N}\tilde{\mu}^k_{an}(t)+\theta_n^k r_n\tilde{s}_n^k(t)\right)|\mathbf
V(t)\right]\nonumber\\
\geq &\lambda^k_n+\epsilon^k_n\label{eqn:proof_ineq1}
\end{align}
Substituting \eqref{eqn:proof_ineq1}
into \eqref{eqn:proof_ineq0}, we have $\Delta (\mathbf V(t))
\leq 2N B-2\sum_{n\in \mathcal N,k\in \mathcal K} \dfrac{\epsilon^k_n}{\theta_n^k}V^k_n(t)
\leq
2N B-2\epsilon\sum_{n\in \mathcal N,k\in \mathcal K}\dfrac{V^k_n(t)}{\theta_n^k} $.
By  Lemma 4.1 of \cite{Georgiadis-Neely-Tassiulas:2006}, we complete the proof.



\end{document}